\documentclass[twocolumn]{autart}    

\usepackage{graphics} 
\usepackage{epsfig} 
\usepackage{amsmath} 
\usepackage{amssymb}  

\usepackage{amsbsy}
\usepackage{mathtools}

\let\theoremstyle\relax 
\let\theoremstyle\relax
\usepackage{amsthm} 
\usepackage{setspace}
\usepackage{ifpdf}
\usepackage{cite}

\usepackage{hyperref}
\hypersetup{
    colorlinks=true,
    citecolor=black,
    filecolor=black,
    linkcolor=black,
    urlcolor=black
}

\usepackage{xcolor}
\usepackage{algorithm}
\usepackage{algpseudocode}

\usepackage{parskip}

\theoremstyle{definition}
\newtheorem{theorem}{Theorem}
\newtheorem{assumption}{Assumption}
\newtheorem{corollary}[theorem]{Corollary}
\newtheorem{lemma}[theorem]{Lemma}

\newtheorem{definition}{Definition}

\newtheorem{remark}{Remark}

\def\RR{{\mathbb R}}
\DeclareMathOperator{\VEC}{vec}

\DeclareMathOperator{\rank}{rank}
\newcommand{\ie}{\textit{i.e.}}
\newcommand{\eg}{\textit{e.g.}}

\begin{document}

\begin{frontmatter}

\title{Data-Driven Model Reduction by Two-Sided Moment Matching
\thanksref{footnoteinfo}} 

\thanks[footnoteinfo]{This paper was not presented at any IFAC 
meeting. Corresponding author J. Mao.}

\author[London]{Junyu Mao}\ead{junyu.mao18@imperial.ac.uk},    
\author[London]{Giordano Scarciotti}\ead{g.scarciotti@imperial.ac.uk}

\address[London]{Department of Electrical and Electronic Engineering, Imperial College London, London, U.K.}  


\begin{abstract}                          
In this brief paper, we propose a time-domain data-driven method for model order reduction by two-sided moment matching for linear systems. 
An algorithm that asymptotically approximates the matrix product $\Upsilon \Pi$ from time-domain samples of the so-called two-sided interconnection is provided. Exploiting this estimated matrix, we determine the unique reduced-order model of order $\nu$, which asymptotically matches the moments at $2 \nu$ distinct interpolation points. Furthermore, we discuss the impact that certain disturbances and data distortions may have on the algorithm. Finally, we illustrate the use of the proposed methodology by means of a benchmark model. 
\end{abstract}

\end{frontmatter}

\section{Introduction}

Model reduction is the problem of approximating some important behaviours (\textit{e.g.} steady-state response) of a given system with a simplified description, \textit{e.g.}, a lower-order model, while some properties of interest (\eg{} stability) are preserved. For linear systems, model reduction techniques have been extensively studied, see the comprehensive survey in \textit{e.g.}, \cite{antoulas2005approximation}. Within those techniques, a popular family is the so-called moment matching methods, in which the resulting reduced-order models match the ``moments'' of the original system at specific frequencies (also known as \textit{interpolation points}), see, \textit{e.g.}, \cite{kimura1986positive, georgiou1999interpolation, byrnes2001generalized}.\\
In a more recent work \cite{astolfi2010model} introduced a time-domain characterization of moments based on the notion of steady state. This characterization brings out two important features: 1) it allows one to extend the moment matching framework to general nonlinear systems in which the classical concept of moments are not conventionally defined, see, \eg, \cite{scarciotti2017nonlinear}; 2) it opens the possibility of constructing families of reduced-order models that asymptotically match the moments over time by exploiting time-domain input/output data from certain system interconnections. In particular, \cite{scarciotti2017data} proposed a data-driven method based on 
the so-called
\textit{direct} interconnection,
obtaining parameterized reduced-order models of order $\nu$ that asymptotically match the moments at $\nu$ interpolation points.\\
In \cite{ionescu2015two} a \textit{model-based} ``\textit{two-sided}'' approach has been proposed. 
The result is a $\nu$-order model that matches moments at two sets of $\nu$ (distinct) interpolation points, thus doubling the number of matched moments with respect to the direct interconnection. In this paper, we seek a time-domain \textit{data-driven} approach to achieve the two-sided moment matching. 
A major difficulty in achieving this objective is that the two-sided interconnection requires data-driven methods for both the direct interconnection (given in \cite{scarciotti2017data}) and the so-called ``\textit{swapped}'' interconnection, which is less intuitive to manipulate.
Our main contribution is the missing piece of the puzzle of two lines of research, namely we provide a data-driven enhancement of \cite{ionescu2015two} in line with \cite{scarciotti2017data} and, respectively, a two-sided enhancement of \cite{scarciotti2017data} in line with \cite{ionescu2015two}. The resulting approach possesses the benefits of both methods: it is a data-driven method (thus it does not require any knowledge of the matrices of the system); it requires a minimum number of time-domain measurements; it has computational complexity which is only a power of the reduced order $\nu$ (thus it is computationally advantageous with respect to model-based methods); and it achieves matching at double the number of interpolation points while maintaining the same order of the reduced model (thus obtaining a better approximation of the original model).
While for reasons of space, we cannot provide a complete review of the model reduction literature at large, for the sake of completeness we mention a few related data-driven moment-based frameworks. The Loewner framework introduced by \cite{mayo2007framework}, constructs interpolants by exploiting \textit{frequency-domain} measurements. While time-domain enhancements are available, it is worth noting that, our approach is natively time-domain, rather than based on the indirect routines that infer frequency-domain information from the time-domain response. This native time-domain nature empowers this data-driven framework to be potentially extended to distributed-parameter systems, and general nonlinear systems where frequency-domain concepts are of great challenge to generalize (and in fact \cite{scarciotti2017data} proposes data-driven methods for nonlinear systems as well). Another framework to be mentioned is the Loewner framework based on system interconnections introduced by \cite{simard2021nonlinear}. This framework is more related to the approach in this brief paper than the classical Loewner framework because it is based on system interconnections. While the method has been extended to nonlinear systems, it still lacks a complete data-driven enhancement. Finally, \cite{scherpen2020data} proposed a method that allows for the exact computation of moments at given interpolation points for \textit{discrete-time} linear systems.
In summary, while the main objective of this brief paper is to finally solve the open problem laid out in \cite{scarciotti2017data} about how to leverage the free parameters of the reduced model to match double the amount of moments in a data-driven context, we also
note that related frameworks do not deal with exactly the same setting (time-domain input-output measurements of continuous-time systems with a minimal number of samples).\\
Some preliminary results have been published in \cite{mao2022model}. Specifically, that paper focused on the swapped interconnection, which is a preliminary result to solve the problem posed in this brief paper. For the sake of space, those results are only briefly recalled in ``Section 2 - Preliminaries''. The present paper focuses on developing methods for the two-sided interconnection, and consequently all results presented herein are novel.\\
The remainder of this paper is organized as follows. In Section~\ref{sec:Moment Estimation via Two System Interconnections} we recall the notion of moment matching, and revisit two data-driven approaches that estimate the matrices $C\Pi$ and $\Upsilon B$, based on the direct and swapped interconnections, respectively. Section~\ref{sec:Data-Driven Moment Matching by A Two-Sided Interconnection} presents a data-driven framework that approximates the product $\Upsilon \Pi$ asymptotically from time-domain samples of the two-sided interconnection. Moreover, we construct the reduced-order model that achieves two-sided moment matching. As an additional result, Section \ref{sec:directEstimationofYPIInverse} presents a result on the direct estimation of $(\Upsilon \Pi)^{-1}$. Furthermore, a discussion of the cases where certain signals are corrupted by noise is also provided in Section \ref{sec:Estimation under Noise Corruption}. Finally, the proposed approach is demonstrated by a building model in Section~\ref{sec:NumericalExampleTwo-sided}.\\
\textbf{Notation}
We use standard notation. $\mathbb{R}$ and $\mathbb{C}$ denote the sets of real numbers and complex numbers respectively. $\mathbb{C}_{0}$ ($\mathbb{C}_{< 0}$) denotes the set of complex numbers with zero (negative) real part. The set of non-negative integers is denoted by $\mathbb{Z}_{\geq 0}$. The identity matrix is denoted by the symbol $I$, and $\sigma(A)$ denotes the spectrum of a square matrix $A$. $A^\top$ and  $\rank (A)$ indicate the transpose and the rank of any matrix $A$, respectively. The symbol $\otimes$ represents the Kronecker product. The operator $\VEC(A)$ indicates the vectorization of a matrix $A \in \mathbb{R}^{n \times m}$, which is the $nm \times 1$ vector obtained by stacking the 
columns of the matrix $A$ one on top of the other.
For a vector $x$, $||x||$ denotes its Euclidean norm. 
$\iota$ denotes the imaginary unit. 

\section{Preliminaries}
\label{sec:Moment Estimation via Two System Interconnections}
In this section, we recall two approaches to estimate moments from input/output data, as proposed in \cite{scarciotti2017data} and \cite{mao2022model} respectively. We end this section by revisiting the model-based two-sided moment matching theory, as proposed in \cite{ionescu2015two}. 

\subsection{Moment Estimation via Direct and Swapped Interconnections}
Consider a linear, single-input single-output (SISO), continuous-time system of the form
\begin{equation} \label{eq:fullOrderSystem}
    \dot{x} = Ax + Bu, \qquad y = Cx,
\end{equation}
with state $x(t) \in \mathbb{R}^{n}$, input $u(t) \in \mathbb{R}$, output $y(t) \in \mathbb{R}$, $A\in \mathbb{R}^{n\times n}$, $B\in \mathbb{R}^{n\times 1}$ and $C\in \mathbb{R}^{1\times n}$. The associated transfer function is $W(s) = C(sI - A)^{-1}B$. Assume system~(\ref{eq:fullOrderSystem}) is minimal, \textit{i.e.}, both controllable and observable. The moments of system~(\ref{eq:fullOrderSystem}) are defined as follows. 

\begin{definition} \label{def:moment}
The $0$-moment of system~(\ref{eq:fullOrderSystem}) at $s_i\in \mathbb{C} \setminus \sigma(A)$ is the complex number $\eta_0(s_i) = W(s_i)$. The $k$-moment of system~(\ref{eq:fullOrderSystem}) at $s_i$ is the complex number $\eta_k(s_i)=\frac{(-1)^{k}}{k!} \left[ \frac{d^k}{ds^k} W(s)\right]_{s = s_i}$, where $k \geq 1$ is an integer. 
\end{definition}


Given two sets of interpolation points $\mathcal{I}_1 = \{s_1, s_2, \cdots, s_\nu\}$ $\subset \mathbb{C} \setminus \sigma(A)$ and  $\mathcal{I}_2 =\{s_{\nu+1}, s_{\nu+2}, \cdots, s_{2\nu}\}$ $\subset \mathbb{C} \setminus \sigma(A)$. 
Suppose that the matrices $S$ and $Q$ are such that $\sigma(S) = \mathcal{I}_1$ and $\sigma(Q) = \mathcal{I}_2$. \cite{astolfi2010model, astolfi2010modelprojection} have shown that the moments at the interpolation points in $\mathcal{I}_1$ and $\mathcal{I}_2$ are in a one-to-one relation with the elements of $C\Pi$ and $\Upsilon B$, respectively, where $\Pi \in \mathbb{R}^{n \times \nu}$ and $\Upsilon \in \mathbb{R}^{\nu \times n}$ are the unique solutions of the Sylvester equations
\begin{subequations} \label{eq:sylvesterEquations}
\begin{align}
A \Pi + BL &= \Pi S,  \label{eq:sylvesterEquationDirect} \\
 Q\Upsilon &= \Upsilon A + RC.  \label{eq:sylvesterEquationSwapped}
\end{align}
\end{subequations}
This one-to-one relation enables $C\Pi$ ($\Upsilon B$, respectively) to play the role of moments as they can be used to construct families of reduced-order models that interpolate the moments at $\mathcal{I}_1$ ($\mathcal{I}_2$, respectively). \cite{astolfi2010model, astolfi2010modelprojection} has also noted that, under additional assumptions, the moments are also in a one-to-one relation with the steady-state responses (provided they exist) of interconnections between the system and certain ``signal generators''. First, consider the signal generator 
\begin{equation} \label{eq:signalGeneratorS}
\dot \omega = S \omega, \qquad \theta = L \omega,
\end{equation}
with $\omega(t) \in \mathbb{R}^{\nu}$, $u(t) \in \mathbb{R}$, $S\in \mathbb{R}^{\nu\times\nu}$ and $L \in \mathbb{R}^{1 \times \nu}$, and the interconnection between this generator and system~(\ref{eq:fullOrderSystem}), namely
\begin{equation} \label{eq:directInterconnection}
    \dot \omega  = S \omega, \qquad \dot{x} = Ax + BL\omega, \qquad y = Cx.
\end{equation}
This interconnection is called \textit{direct interconnection}. Under certain technical assumptions, the moments at $\mathcal{I}_1$ are in a one-to-one relation with the steady state $y_{ss}$ of \eqref{eq:directInterconnection}. Next, consider the signal generator 
\begin{equation} \label{eq:signalGenerator}
    \dot{\varpi} = Q \varpi + R \kappa, \qquad d = \varpi + \Upsilon x,
\end{equation}
with $\varpi(t) \in \mathbb{R}^{\nu}$, $\kappa(t) \in \mathbb{R}$, $d(t) \in \mathbb{R}^{\nu}$, $Q \in \mathbb{R}^{\nu \times \nu}$, $R \in \mathbb{R}^{\nu \times 1}$ and $\Upsilon \in \mathbb{R}^{\nu \times \nu}$, and the interconnection between this generator and system~(\ref{eq:fullOrderSystem}), namely
\begin{equation} \label{eq:swappedInterconnection}
    \dot{x} = Ax + Bu, \quad \dot \varpi = Q \varpi + R C x, \quad d = \varpi + \Upsilon x, 
\end{equation}
with $u = \delta_0$, where $\delta_0$ indicates the Dirac-delta. This interconnection is called \textit{swapped interconnection}. Under certain technical assumptions, the moments at $\mathcal{I}_2$ are in a one-to-one relation with the steady state $d_{ss}$ of \eqref{eq:swappedInterconnection}.\\
The one-to-one relations between moments, matrices $C \Pi$ and $\Upsilon B$, and steady states $y_{ss}$ and $d_{ss}$ underpin two approaches (proposed in \cite{scarciotti2017data} and \cite{mao2022model} respectively) to estimate the matrices $C \Pi$ and $\Upsilon B$ from time-domain measurements. In Theorems \ref{theorem:computeCPI} and \ref{theorem:computeYB} below we recall these two approaches. To streamline the presentation, we first introduce the following assumptions. 

\begin{assumption} \label{assumption:minimal}
For the signal generator~(\ref{eq:signalGeneratorS}), the pair $(S,L)$ is observable. For the signal generator~(\ref{eq:signalGenerator}), the pair $(Q,R)$ is controllable. 
\end{assumption}

\begin{assumption} \label{assumption:SSPI}
System~(\ref{eq:fullOrderSystem}) is asymptotically stable, \textit{i.e.}, $\sigma(A) \subset \mathbb{C}_{<0}$. The matrix $S$ has simple eigenvalues with $\sigma(S) \subset \mathbb{C}_0$. The pair $(S,\omega_0)$ is excitable\footnote{See \cite{padoan2016geometric} for the definition of excitable pair.}.
\end{assumption}

\begin{assumption} \label{assumption:SSY}
System~(\ref{eq:fullOrderSystem}) is asymptotically stable. The matrix $Q$ has simple eigenvalues with $\sigma(Q) \subset \mathbb{C}_0$. The initial condition $x(0)$ is $0$.
\end{assumption}
Note that under the asymptotic stability of system (\ref{eq:fullOrderSystem}), the assumption on the zero initial condition of the state $x(t)$ is without loss of generality. 

\begin{theorem}[\cite{scarciotti2017data}] \label{theorem:computeCPI}
Consider the interconnection (\ref{eq:directInterconnection}) and suppose Assumptions \ref{assumption:minimal} and \ref{assumption:SSPI} hold. Define the time-snapshot matrices $\tilde{R}_k \in \mathbb{R}^{w \times \nu}$ and $\tilde{\Gamma}_k \in \mathbb{R}^{w}$, with $w \geq \nu$, as
$$
\tilde{R}_k = \left[\omega(t_{k- w + 1}) \; \hdots \; \omega(t_{k-1}) \; \omega(t_{k}) \right]^\top
$$
and 
$$
\tilde{\Gamma}_k = 
\left[
y(t_{k- w + 1}) \; \hdots \;  y(t_{k-1}) \; y(t_{k}) 
\right]^\top,
$$
respectively. Then,
\begin{equation} \label{eq:computeCPI}
    \VEC(\widetilde{C \Pi}_k) := (\tilde{R}_k^\top \tilde{R}_k)^{-1} \tilde{R}_k^\top \tilde{\Gamma}_k,
\end{equation}
is an online estimate of $C \Pi$, namely there exists a sequence $\{t_k\}$ such that $\widetilde{C \Pi}_k$ is well-defined for all $k$ and $\lim_{k \to \infty} \widetilde{C \Pi}_k = C \Pi$.
\end{theorem}

\begin{theorem}[\cite{mao2022model}] \label{theorem:computeYB}
Consider the interconnection~(\ref{eq:swappedInterconnection})  with $\varpi(0) = 0$ and $u = \delta_0$. Suppose Assumptions \ref{assumption:minimal} and \ref{assumption:SSY} hold. Then, 
\begin{equation} \label{eq:computeMomentSingalPoint}
    \widetilde{\Upsilon B}_i := e^{-Q t_i} \varpi(t_i)
\end{equation}
is an online estimate of $\Upsilon B$, namely there exists a sequence $\{t_i\}$ such that  
$
    \lim_{i \to \infty} \widetilde{\Upsilon B}_i = \Upsilon B. 
$
Moreover, the transient error $\epsilon_{\Upsilon B}^i := \Upsilon B - \widetilde{\Upsilon B}_i$ at time $t_i$ is
\begin{equation} \label{eq:errorYB}
    \epsilon_{\Upsilon B}^i = e^{-Qt_i}\Upsilon e^{At_i} B.
\end{equation}
\end{theorem}

\subsection{Two-Sided Moment Matching}
Consider the signal generators (\ref{eq:signalGeneratorS}) and (\ref{eq:signalGenerator}), and the interconnection between these generators and system~(\ref{eq:fullOrderSystem}),
\begin{equation} \label{eq:twosidedInterconnection}
\begin{split} 
     \dot{\omega} = S \omega, \qquad \qquad \,
     \dot{x} = Ax + BL\omega, \,\, \\
     \dot \varpi = Q \varpi + R C x,  \qquad
     d = \varpi + \Upsilon x.
\end{split}
\end{equation}
This interconnection is called \textit{two-sided interconnection} and is schematically represented in Fig.~\ref{fig:twosidedInterconnections}. In the following we briefly recall the theory of two-sided moment matching as proposed in \cite{ionescu2015two}. To ensure the uniqueness of the solution for \eqref{eq:sylvesterEquations} and to match the moments at as many interpolation points as possible, we formalize the required assumption compactly.
\begin{assumption} \label{assumption:invertible}
The matrices $S, Q$ and $A$ have no common eigenvalues, \ie,
$\sigma(S) \cap \sigma(A) = \emptyset,  \sigma(Q) \cap \sigma(A) = \emptyset, \sigma(S) \cap \sigma(Q) = \emptyset$.
\end{assumption}

In addition, a technical assumption that ensures the existence of the two-sided reduced-order model is also introduced as follows. 
\begin{assumption} \label{assumption:invertible2}
The matrix $\Upsilon \Pi$ is invertible. 
\end{assumption}

\begin{figure}[tpb]
    \centering
    \includegraphics[scale=0.24]{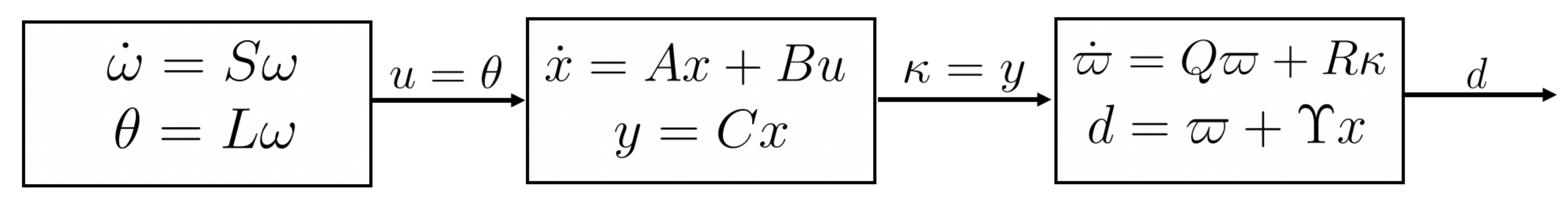}
    \caption{Diagrammatic illustration of the \textit{two-sided interconnection}.}
    \label{fig:twosidedInterconnections}
\end{figure}

\begin{theorem}[\cite{ionescu2015two}] \label{theorem:reducedOrderModelforTwo-sided}
Consider system (\ref{eq:fullOrderSystem}). Let $\Pi$ and  $\Upsilon$ be the unique solutions of \eqref{eq:sylvesterEquations}.  Suppose Assumptions~\ref{assumption:minimal}, \ref{assumption:invertible} and \ref{assumption:invertible2} hold.  
Then, there exists a unique ($\nu$-order) model that matches the moments of system (\ref{eq:fullOrderSystem}) at $\sigma(S)$ and $\sigma(Q)$ simultaneously, namely
\begin{equation} \label{eq:reducedOrderModelTwo-sided1}
    \dot{\xi} = (Q - RH)\xi + \Upsilon B u, \qquad \psi = H \xi,
\end{equation}
if and only if $H = C \Pi (\Upsilon \Pi)^{-1}$, or, equivalently via similarity transformation,
\begin{equation} \label{eq:reducedOrderModelTwo-sided2}
    \dot{\xi} = (S - GL)\xi + G u, \qquad \psi = C \Pi \xi,
\end{equation}
if and only if $G = (\Upsilon \Pi)^{-1} \Upsilon B$.

\end{theorem}

Theorem \ref{theorem:reducedOrderModelforTwo-sided} presents sufficient and necessary conditions for a model of order $\nu$ to match the moments of system (\ref{eq:fullOrderSystem}) at $2\nu$ distinct interpolation points characterized by $S$ and $Q$ jointly. 
Given the matrices $S$ and $Q$, which characterize the sets of interpolation points, it can be noted that the parametrization (\ref{eq:reducedOrderModelTwo-sided1}), or equivalently (\ref{eq:reducedOrderModelTwo-sided2})\footnote{For brevity, in the remainder of this paper we will only refer to (\ref{eq:reducedOrderModelTwo-sided1}) as a representative of the unique model in Theorem \ref{theorem:reducedOrderModelforTwo-sided} due to the equivalence between (\ref{eq:reducedOrderModelTwo-sided1}) and (\ref{eq:reducedOrderModelTwo-sided2}).}, is fully determined by three ``moment-related'' matrices $C \Pi, \Upsilon B$ and $\Upsilon \Pi$. Since the data-driven estimation of $C \Pi$ and $\Upsilon B$ has been addressed as recalled in Section \ref{sec:Moment Estimation via Two System Interconnections}, the crucial problem of estimating $\Upsilon \Pi$, or $(\Upsilon\Pi)^{-1}$, from data
remains to be the only open problem.

\begin{remark}
One of the main difficulties of this problem is that the signal $d$ is not available for measurement. In fact, $d$ depends on $x$ and $\Upsilon$, which are both unknown in an input-output data setting.
\end{remark}

\section{Data-Driven Moment Matching via Two-Sided Interconnection}
\label{sec:Data-Driven Moment Matching by A Two-Sided Interconnection}



In this section we provide an algorithm for the estimation of $\Upsilon\Pi$ (and $(\Upsilon\Pi)^{-1}$ in Section~\ref{sec:directEstimationofYPIInverse}) from input-output data. To this end, we first obtain a preliminary result by assuming availability of the signal $d$. Then, by exploiting a surrogate system, we remove this requirement. Finally, in Section~\ref{sec:Estimation under Noise Corruption} we discuss the effect of noise on the algorithm.\\
Consider the two-sided interconnection~(\ref{eq:twosidedInterconnection}). The steady-state of $x(t)$ and the state $\omega(t)$ define a global invariant manifold described by $
\mathcal{M}=\{(x, \omega) \in \mathbb{R}^{n+\nu}:$ $x=\Pi \omega\}$, see \cite{astolfi2010model}. In particular, for all $t \geq 0$, it holds that
\begin{equation} \label{eq:xandOmega}
    x(t) = \Pi \omega(t) + e^{At}(x(0) - \Pi \omega(0)).
\end{equation}
Based on (\ref{eq:xandOmega}), the second equation in (\ref{eq:signalGenerator}) is rewritten as
\begin{equation} \label{eq:onlineEstimationofYPi0}
    d(t) = \varpi(t) +  \Upsilon x(t) =  \varpi(t) +  \Upsilon \Pi \omega(t) + \epsilon_1(t),
\end{equation}
where $\epsilon_1(t) = \Upsilon e^{At} (x(0) - \Pi \omega(0))$ is an exponentially decaying error signal. This yields
\begin{equation} \label{eq:onlineEstimationofYPi1}
    \lim_{t \to \infty} \left(d(t) - \varpi(t) - \Upsilon \Pi \omega(t)\right) = 0.
\end{equation}
We are now in a position to present a (preliminary) result that allows us to asymptotically estimate the matrix $\Upsilon \Pi$ through time-domain measurements by unrealistically having access to the signal $d$.  
\begin{lemma} \label{theorem:computeYPI}
Consider the interconnection (\ref{eq:twosidedInterconnection}) with $\varpi(0) = 0$. Suppose Assumptions~\ref{assumption:minimal}, \ref{assumption:SSPI}, \ref{assumption:SSY} and \ref{assumption:invertible} hold. Let the time-snapshot matrices $O_k \in \mathbb{R}^{p \nu \times \nu^2}$ and $P_k \in \mathbb{R}^{p \nu}$, with $p \geq \nu$, be defined as 
$$
O_k \!\!=\!\!\! 
\begin{bmatrix}
\omega(t_{k-p+1})^{\top} \!\otimes\! I_\nu \\ 
\vdots \\ 
 \omega(t_{k-1})^{\top} \otimes I_\nu \\ 
 \omega(t_{k})^{\top} \otimes I_\nu \\ 
\end{bmatrix}\!\!,\, P_k \!\!=\!\!\!
\begin{bmatrix}
d(t_{k-p+1}) \!-\! \varpi(t_{k-p+1}) \\
\vdots \\
d(t_{k-1}) - \varpi(t_{k-1}) \\
d(t_{k}) - \varpi(t_{k}) \\
\end{bmatrix}\!\!, 
$$
respectively. Then, 
\begin{equation} \label{eq:computeYPI}
    \VEC(\widetilde{\Upsilon \Pi}_k) := (O_k^\top O_k)^{-1} O_k^\top P_k,
\end{equation}

is an online estimate of $\Upsilon \Pi$, namely there exists a sequence $\{t_k\}$ such that $\widetilde{\Upsilon \Pi}_k$ is well-defined for all $k$ and
\begin{equation} \label{eq:computeYPIProperty}
    \lim_{k \to \infty} \widetilde{\Upsilon \Pi}_k = \Upsilon \Pi.
\end{equation}

\end{lemma}

\begin{proof}
By equation \eqref{eq:onlineEstimationofYPi0}, we rewrite $P_k$ as 
$$
P_k = 
\underbrace{
\begin{bmatrix}
\Upsilon \Pi \omega(t_{k-p+1}) \\
\vdots \\
\Upsilon \Pi \omega(t_{k-1}) \\
\Upsilon \Pi \omega(t_{k}) \\
\end{bmatrix}}_{\displaystyle D_k}
+ \underbrace{\begin{bmatrix}
\epsilon_1(t_{k-p+1}) \\
\vdots \\
\epsilon_1(t_{k-1}) \\
\epsilon_1(t_{k}) \\
\end{bmatrix}}_{\displaystyle E_k}.
$$
Applying the properties\footnote{$\VEC(\mathcal{A} \mathcal{B} \mathcal{C}) = (\mathcal{C}^\top \otimes \mathcal{A}) \VEC(\mathcal{B})$, with $\mathcal{A}, \mathcal{B}$ and $\mathcal{C}$ of compatible dimensions.} of the vectorization operator on the term $D_k$ yields 
\begin{equation} \label{eq:PkRewritten}
    P_k = O_k \VEC (\Upsilon \Pi) + E_k. 
\end{equation}
By Assumption~\ref{assumption:SSPI}, \cite[Lemma 3]{padoan2016geometric} has shown that there exists a sequence $\{t_k\}$ with $\lim_{k \to \infty} t_k = \infty$ such that the rank of the matrix $\tilde{R}_k$ (as defined in Theorem~\ref{theorem:computeCPI}) equals $\nu$. Obverse that $O_k = \tilde{R}_k^\top \otimes I_\nu$. Applying a property\footnote{$\rank(\mathcal{A} \otimes \mathcal{B}) = \rank(\mathcal{A}) \rank(\mathcal{B})$.} of the Kronecker product yields that $\rank(O_k) = \rank(\tilde{R}_k^\top) \rank(I_\nu) = \nu^2$, that is, $O_k$ has full column rank. This directly implies the invertibility of $O_k^\top O_k$, and $\widetilde{\Upsilon \Pi}_k$ is therefore well-defined. Moreover, by Assumption~\ref{assumption:SSPI}, the signal $\omega(t) = e^{St}\omega(0)$ is bounded. Then, it follows directly that $(O_k^\top O_k)^{-1} O_k^\top$ is bounded in each entry. Substituting equation \eqref{eq:PkRewritten} into equation \eqref{eq:computeYPI} yields 
\begin{align*} 
    \lim_{t_k \to \infty}& \VEC (\widetilde{\Upsilon \Pi}_k) = \lim_{t_k \to \infty}  (O_k^\top O_k)^{-1} O_k^\top P_k \\
    &= \lim_{t_k \to \infty}  (O_k^\top O_k)^{-1} O_k^\top  (O_k \VEC (\Upsilon \Pi) + E_k) \\
    &= \lim_{t_k \to \infty} \left( \VEC (\Upsilon \Pi) + (O_k^\top O_k)^{-1} O_k^\top E_k \right) \\
    &= \VEC(\Upsilon \Pi) + \lim_{t_k \to \infty} (O_k^\top O_k)^{-1} O_k^\top E_k, 
\end{align*} 
Notice that, under Assumption~\ref{assumption:SSPI}, $\lim_{t_k \to \infty} E_k = 0$. Hence, the proof is complete. 
\end{proof}


Note that the result of Lemma~\ref{theorem:computeYPI} requires knowledge of the signal $d$, which is not normally available from experimental measurements. Therefore, here we present an instrumental observation which we exploit in the following to estimate the output signal $d$ using $\Upsilon B$, which we can estimate using Theorem~\ref{theorem:computeYB}. 

\begin{remark} \label{remark:dSystem}
    Note that the evolution of signal $d(t)$ can be expressed by the system described by
\begin{equation}  \label{eq:dSystem}
    \dot{d} = Qd + \Upsilon BL\omega, \qquad d(0) = 0. 
\end{equation}
This directly follows from the interconnection equation (\ref{eq:twosidedInterconnection}) and the Sylvester equation (\ref{eq:sylvesterEquationSwapped}), namely
    \begin{align*} 
        \dot{d} &= \dot{\varpi} + \Upsilon \dot{x} =  Q \varpi + RCx + \Upsilon \left(Ax + Bu \right) \\
        &=  Q \varpi + \underbrace{\left(RC + \Upsilon A\right)}_{= \,Q\Upsilon}x + \Upsilon Bu =  Q (\varpi + \Upsilon x) +  \Upsilon Bu \\
        &=  Q d +  \Upsilon Bu =  Q d +  \Upsilon BL\omega, 
    \end{align*}
with the initial condition $d(0) = \varpi(0) + \Upsilon x(0) = 0$. 
\end{remark}

This observation implies that, if we are provided with the estimation $\widetilde{\Upsilon B}_i$ computed using Theorem~\ref{theorem:computeYB}, we can construct a \textit{surrogate system} described by
\begin{equation} \label{eq:surrogateSystem}
    \dot{\hat{d}} = Q\hat{d} + \widetilde{\Upsilon B}_i L \omega, \qquad \hat{d}(0) = 0,
\end{equation}
with the property that there exists a sequence $\{t_i\}$ such that $
    \lim_{i \to \infty}\hat{d}(t) = d(t), \forall t \geq 0.$\\
The surrogate system \eqref{eq:surrogateSystem} allows using the signal $\hat{d}$ as an approximate replacement of $d$ in Lemma~\ref{theorem:computeYPI}. Note, however, that by substituting $d(t)$ with $\hat{d}(t)$, the result of Lemma~\ref{theorem:computeYPI} holds if and only if $\widetilde{\Upsilon B}_i = \Upsilon B$, which holds only when $t_i \to \infty$. For any finite time $t_i \neq \infty$, it can be seen that the error \eqref{eq:errorYB} is non-zero. 
As a direct result of $\widetilde{\Upsilon B}_i \neq \Upsilon B$, there exists a discrepancy between the actual and the estimated values of $d(t)$, which indicates that the convergence to $\Upsilon \Pi$ in \eqref{eq:computeYPIProperty} no longer holds. Nevertheless, in what follows we show that, even in this case, the estimation of $\Upsilon \Pi$ obtained using $\hat{d}(t)$ lies in a bounded neighbourhood of its exact value. We start with the following lemma. 

\begin{lemma} \label{lemma:errorDBounded}
Consider the signals $d(t)$ and $\hat{d}(t)$ described by the systems \eqref{eq:dSystem} and \eqref{eq:surrogateSystem}, respectively. Suppose Assumptions~\ref{assumption:SSPI}, \ref{assumption:SSY} and \ref{assumption:invertible} hold. Then, there exists $\beta > 0$ such that $||d(t) - \hat{d}(t)||_\infty \leq \beta$ for all $t \geq 0$.
\end{lemma}
\begin{proof}
By \eqref{eq:dSystem} and \eqref{eq:surrogateSystem}, we notice that the evolution of the error signal $\epsilon_d(t) := d(t) - \hat{d}(t)$ is governed by the interconnected system
\begin{equation} 
    \dot{\omega} = S \omega, \qquad  \dot{\epsilon_d} = Q\epsilon_d + \epsilon_{\Upsilon B}^i L \omega,
\end{equation}
with $\omega(0) = \omega_0$, $\epsilon_d(0) = 0$ and $\epsilon_{\Upsilon B}^i$ defined in (\ref{eq:errorYB}). The above system can be rewritten in a compact form as
\begin{equation} \label{eq:EpsilonDsystem}
\begin{bmatrix}
    \dot{\epsilon_d} \\
    \dot{\omega}
\end{bmatrix}
=
\underbrace{
\begin{bmatrix}
    Q & \epsilon_{\Upsilon B}^i L \\
    0 & S 
\end{bmatrix}
}_{\displaystyle\mathcal{A}_d}
\begin{bmatrix}
    \epsilon_d \\
    \omega
\end{bmatrix}.     
\end{equation}
Note that $\mathcal{A}_d \in \RR^{\nu \times \nu}$ is a block upper triangular matrix, the eigenvalues of which are exactly the eigenvalues of its diagonal blocks (see \cite[Lemma 7.1.1]{golub2013matrix}), \ie, $\sigma(\mathcal{A}_d) = \sigma(Q) \cup \sigma(S)$. Since by Assumptions \ref{assumption:SSPI} and \ref{assumption:SSY}, $S$ and $Q$ only have simple eigenvalues on the imaginary axis, and by Assumption \ref{assumption:invertible}, $\sigma(Q) \cap \sigma(S) = \emptyset$, it is obvious to see that the matrix $\mathcal{A}_d$ also has simple eigenvalues with $\sigma(\mathcal{A}_d) \subset \mathbb{C}_0$, which implies that system \eqref{eq:EpsilonDsystem} is (marginally) stable. As a direct result, the signal $\epsilon_d(t)$ is bounded for all $t \geq 0$, namely there exists a real number $\beta$ such that $||d(t) - \hat{d}(t)||_\infty \leq \beta, \forall t \geq 0$. 
\end{proof}

\begin{theorem} \label{theorem:computeYPIApproximate}
Consider the interconnection (\ref{eq:twosidedInterconnection}) with $\varpi(0) = 0$. Suppose Assumptions~\ref{assumption:minimal}, \ref{assumption:SSPI}, \ref{assumption:SSY} and \ref{assumption:invertible} hold. Let $\widetilde{\Upsilon B}_i$ be an estimate of $\Upsilon B$. Let $O_k$ be as in Lemma~\ref{theorem:computeYPI} and the matrix $\widehat{P}_k \in \mathbb{R}^{p \nu}$, with $p \geq \nu$, be defined as
$$
\widehat{P}_k = 
\begin{bmatrix}
\hat{d}(t_{k-p+1}) - \varpi(t_{k-p+1}) \\
\vdots \\
\hat{d}(t_{k-1}) - \varpi(t_{k-1}) \\
\hat{d}(t_{k}) - \varpi(t_{k}) \\
\end{bmatrix},
$$
with $\hat{d}(t)$ obtained using the surrogate system~\eqref{eq:surrogateSystem}. Then, there exists a sequence $\{t_k\}$ with $\lim_{k \to \infty} t_k = \infty$ such that, for the approximate estimate defined as
\begin{equation} \label{eq:computeYPIApproximate}
    \VEC(\widehat{\Upsilon \Pi}_k) := (O_k^\top O_k)^{-1} O_k^\top \widehat{P}_k,
\end{equation} there exists $\gamma > 0$ such that as $k$ approaches infinity,
\begin{equation} \label{eq:errorBound}
     ||\widehat{\Upsilon \Pi}_k  - \Upsilon \Pi ||_{\infty} \leq \gamma. 
\end{equation}
\end{theorem}

\begin{proof}
We have to prove that as $t_k$ approaches infinity, the error $|| \VEC(\widehat{\Upsilon \Pi}_k) - \VEC(\Upsilon \Pi)||_\infty$ is bounded by a certain constant $\gamma$. We note that 
\begin{equation} \label{eq:diffDandDhat}
    \widehat{P}_k = P_k + 
    \underbrace{
    \begin{bmatrix}
        \hat{d}(t_{k-p+1}) - d(t_{k-p+1}) \\
        \vdots \\
        \hat{d}(t_{k-1}) - d(t_{k-1}) \\
        \hat{d}(t_{k}) - d(t_{k})
    \end{bmatrix}}_{\displaystyle Z_k}. 
\end{equation}
Then it follows from \eqref{eq:diffDandDhat} and \eqref{eq:PkRewritten} that 
$$
    \widehat{P}_k = O_k \VEC (\Upsilon \Pi) + E_k + Z_k. 
$$
Substituting the above equation into \eqref{eq:computeYPIApproximate}, yields
\begin{equation*} \label{eq:infinityBound0}
    \begin{split}  
    & || \VEC (\widehat{\Upsilon \Pi}_k) - \VEC(\Upsilon \Pi)||_\infty \\
    &= || (O_k^\top O_k)^{-1} O_k^\top \widehat{P}_k - \VEC(\Upsilon \Pi) ||_\infty \\
    &= || (O_k^\top O_k)^{-1} O_k^\top  (O_k \VEC (\Upsilon \Pi)  \! + \! E_k \! - \! Z_k) \! - \! \VEC(\Upsilon \Pi) ||_\infty \\
    &= || (O_k^\top O_k)^{-1} O_k^\top  ( E_k - Z_k) ||_\infty.
\end{split}
\end{equation*}
As $t_k$ tends to $\infty$, the term $E_k$ goes to $0$ and it follows that
\begin{equation} \label{eq:infinityBound}
    \begin{split}  
     || \VEC  (\widehat{\Upsilon \Pi}_k)  - \VEC(\Upsilon \Pi)||_\infty 
    \leq  ||(O_k^\top O_k)^{-1} O_k^\top||_\infty || Z_k ||_\infty. 
\end{split}
\end{equation}
In the proof of Lemma~\ref{theorem:computeYPI}, it has been shown that the matrix $(O_k^\top O_k)^{-1} O_k^\top$ is bounded in each entry, that is, there exists $\alpha > 0$ such that $||(O_k^\top O_k)^{-1} O_k^\top||_\infty \leq \alpha, \forall t_k \geq 0$. In addition, Lemma~\ref{lemma:errorDBounded} shows that there exists $\beta > 0$ such that $||d(t_k) - \hat{d}(t_k)||_\infty \leq \beta$ for all $t_k \geq 0$, from which it follows directly that $|| Z_k ||_\infty \leq \beta, \forall t_k \geq 0$. Then, by the inequality in \eqref{eq:infinityBound}, simply selecting $\gamma = \alpha \beta$ yields that condition \eqref{eq:errorBound} holds. 
\end{proof}

In what follows, Algorithm~\ref{alg:On-LineMomentEstimationTwo-sided} details the associated computational procedures, based on Theorems \ref{theorem:computeCPI} and \ref{theorem:computeYPIApproximate}, used to estimate online the matrices $C \Pi$ and $\Upsilon \Pi$ \textit{simultaneously}.

\begin{algorithm}
  \caption{On-Line Estimation of $C \Pi$ and $\Upsilon \Pi$.}
  \label{alg:On-LineMomentEstimationTwo-sided}
  \begin{algorithmic}[1]
    \State{\textbf{Input:} time sequence $\{t_k\}$, $w = p = \nu$, \textit{sufficiently} small tolerances $\eta_{C \Pi} > 0$ and $\eta_{\Upsilon \Pi} > 0$, a \textit{sufficiently} large $k \in \mathbb{Z}_{\geq 0}$, estimation $\widetilde{\Upsilon B}_i$}
    \While{1}
    \State{/*** Estimate $C \Pi$ ***/}
    \State{construct data matrices $\tilde{R}_k$ and $\tilde{\Gamma}_k$}
    \If{rank$(\tilde{R}_k) = \nu$}
    \State{compute $\widetilde{C \Pi}_k$ \Comment{see (\ref{eq:computeCPI})}}
    \Else
    \State{increase $w$}
    \EndIf

    \State{}
    \State{/*** Estimate $\Upsilon \Pi$ ***/}
    \State{generate $\hat{d}(t_k)$ using surrogate system (\ref{eq:surrogateSystem})}
    \State{construct data matrices $O_k$ and $\widehat{P}_k$}
    \If{rank$(O_k) = \nu^2$}
    \State{compute $\widehat{\Upsilon \Pi}_k$ \Comment{see (\ref{eq:computeYPIApproximate})}}
    \Else
    \State{increase $p$}

    \EndIf
    
    \State{}
    \State{/*** Stopping Condition ***/}
    \If{$|| \widetilde{C \Pi}_k -  \widetilde{C \Pi}_{k-1} || \leq (t_k - t_{k-1})\eta_{C \Pi}$ and $\qquad \qquad \qquad $ $|| \widehat{\Upsilon \Pi}_k -  \widehat{\Upsilon \Pi}_{k-1} || \leq (t_k - t_{k-1})\eta_{\Upsilon \Pi}$} 
        \State{\textbf{break}}
    \EndIf
    \State{$k = k + 1$}
    \EndWhile
    \State{\textbf{Return:} $\widetilde{C \Pi}_{k}, \widehat{\Upsilon \Pi}_{k}$}
  \end{algorithmic}
\end{algorithm}

\medskip

\begin{remark}
The estimations of the matrices $\Upsilon B$, $C \Pi$ and $\Upsilon \Pi$ are split across two experiments. In the first one the swapped interconnection~\eqref{eq:swappedInterconnection} is employed to obtain the approximation $\widetilde{\Upsilon B}_i$ using Theorem~\ref{theorem:computeYB}. In the second the two-sided interconnection~\eqref{eq:twosidedInterconnection} is used to obtain the estimations $\widetilde{C \Pi}_k$ and $\widehat{\Upsilon \Pi}_k$ simultaneously, by Algorithm~\ref{alg:On-LineMomentEstimationTwo-sided}. 
\end{remark}




With the estimates $\widetilde{\Upsilon B}_i, \widetilde{C \Pi}_k$ and $\widehat{\Upsilon \Pi}_k$ obtained from the experiments, we are now able to determine an approximate data-driven two-sided reduced-order model.  The \textit{approximate reduced-order model of system~(\ref{eq:fullOrderSystem}) at the interpolation pair $(\sigma(S), \sigma(Q))$, at time $t_k$, and at the estimate $\widetilde{\Upsilon B}_i$}, is defined as
\begin{equation} \label{eq:dataDrivenReducedOrderModel}
    \dot{\xi} = (Q - RH_k)\xi + \widetilde{\Upsilon B}_i u, \qquad \psi = H_k \xi,
\end{equation}
with $H_k = \widetilde{C \Pi}_k (\widehat{\Upsilon \Pi}_k)^{-1}$, provided $\widehat{\Upsilon \Pi}_k$ is invertible.

\begin{remark}
If Assumption \ref{assumption:invertible2} holds (as assumed in \cite{ionescu2015two}), then there exists a sufficiently large $\bar{k}$ such that $\widehat{\Upsilon \Pi}_k$ is invertible for all $k \geq \bar{k}$. Hence, from now on, we assume the matrix $\widehat{\Upsilon \Pi}_k$ is invertible without loss of generality.
\end{remark}

\begin{remark}
The approximate model in \eqref{eq:dataDrivenReducedOrderModel} mismatches the unique two-sided reduced-order model \eqref{eq:reducedOrderModelTwo-sided1} due to the estimation errors for $\Upsilon B$, $C \Pi$ and $\Upsilon \Pi$. This mismatch can be decreased by following these two steps: first, one can obtain a more accurate estimation of $\Upsilon B$ by performing an experiment of longer duration by Theorem~\ref{theorem:computeYB}; then, smaller tolerances can be selected in Algorithm~\ref{alg:On-LineMomentEstimationTwo-sided} to obtain better approximations of $C \Pi$ and $\Upsilon \Pi$.
\end{remark}

Considering that the matrix inversion operation may amplify the estimation errors, in the next section we present a result that allows estimating the inverse $(\Upsilon \Pi)^{-1}$ directly.

\subsection{Direct Estimation of $(\Upsilon \Pi)^{-1}$} \label{sec:directEstimationofYPIInverse}
Observe that the information of $\Upsilon \Pi$ is not necessarily needed if we are able to directly obtain the matrix $(\Upsilon \Pi)^{-1}$. In this section, we present a result that allows us to estimate the inverse $(\Upsilon \Pi)^{-1}$ directly, without the need of computing the matrix inversion on the obtained estimation of $\Upsilon \Pi$ as in \eqref{eq:dataDrivenReducedOrderModel}.

\begin{corollary} \label{theorem:computeYPIInvese}
Consider the interconnection (\ref{eq:twosidedInterconnection}) with $\varpi(0) = 0$. Suppose Assumptions~\ref{assumption:minimal}, \ref{assumption:SSPI}, \ref{assumption:SSY},  \ref{assumption:invertible} and \ref{assumption:invertible2} hold. Let the time-snapshot matrices $M_k \in \RR^{p \nu \times \nu^2}$ and $N_k \in \RR^{p \nu}$, with $p \geq \nu$, be defined as
$$
M_k = 
\begin{bmatrix}
(d(t_{k-p+1})^\top - \varpi(t_{k-p+1})^\top) \otimes I_\nu \\
\vdots \\
(d(t_{k-1})^\top - \varpi(t_{k-1})^\top) \otimes I_\nu \\
(d(t_{k})^\top - \varpi(t_{k})^\top) \otimes I_\nu \\
\end{bmatrix},
$$
and
$$
N_k =
\begin{bmatrix}
\omega(t_{k-p+1})^\top & \cdots & \omega(t_{k-1})^\top & \omega(t_{k})^\top
\end{bmatrix}^\top,
$$

respectively. Assuming the full column rank of the matrix $M_k$, then 
\begin{equation} \label{eq:computeYPIInverse}
    \VEC(\widetilde{\Sigma}_k) := (M_k^\top M_k)^{-1} M_k^\top N_k,
\end{equation}

is an online estimate of $(\Upsilon \Pi)^{-1}$, namely there exists a sequence $\{t_k\}$ such that
\begin{equation*} \label{eq:computeYPIInverseProperty}
    \lim_{k \to \infty} \widetilde{\Sigma}_k = (\Upsilon \Pi)^{-1}.
\end{equation*}
\end{corollary}

\begin{proof}
We observe that by multiplying $(\Upsilon \Pi)^{-1}$ on both sides of \eqref{eq:onlineEstimationofYPi0} and rearranging the equation, we have
$$
    (\Upsilon \Pi)^{-1}(d(t) - \varpi(t)) =  \omega(t) + (\Upsilon \Pi)^{-1}\epsilon_1(t),
$$
As $\epsilon_1(t) = \Upsilon e^{At} (x(0) - \Pi \omega(0))$ is exponentially decaying in time, we obtain
$$
\lim_{t \to \infty} \left((\Upsilon \Pi)^{-1}(d(t) - \varpi(t)) - \omega(t)\right) = 0.
$$
Based on the above equations, the rest of the proof follows the same procedure used in the proof of Lemma~\ref{theorem:computeYPI}. The details are hence omitted for the sake of space. 
\end{proof}

We highlight that, by obviating the need for computing the matrix inversion,  this last approach has several numerical advantages: 1) the computational complexity is reduced; 2) the approach is numerically more precise due to the fact that the matrix inversion is rather sensitive to small perturbations (thus  estimation errors could be amplified if the inversion is performed); 3) the method is numerically more stable in the subsequent stage of constructing the reduced-order model, as there is now no technical requirement for the invertibility of the online estimation of $\Upsilon \Pi$ (which may be singular for some $k$ even though $\Upsilon \Pi$ is not).\\
Corollary~\ref{theorem:computeYPIInvese} can be modified to replace $d$ with $\hat{d}$, obtaining a similar bound as (\ref{eq:errorBound}). Then, one can easily adapt Algorithm~\ref{alg:On-LineMomentEstimationTwo-sided} and the corresponding data-driven reduced-order model \eqref{eq:dataDrivenReducedOrderModel} based on the estimation $\widetilde{\Sigma}_k$ in \eqref{eq:computeYPIInverse} with minor efforts. 

\subsection{Estimation under Noise Corruption} \label{sec:Estimation under Noise Corruption}
In practice, it is very common for data-driven methods to be affected by different kinds of noises. For instance, one may have noisy measurements of $\omega(t)$ and $\varpi(t)$, which could degrade the estimation quality. We first observe that the estimate \eqref{eq:computeYPI} is a least squares (LS) solution of equation \eqref{eq:onlineEstimationofYPi1} by augmenting sufficiently informative data. This provides an optimal (in the sense of maximum likelihood) estimation for the case where $\varpi(t)$ is corrupted by additive white noises. In a more general case where both $\varpi(t)$ and $\omega(t)$ are affected under white noises, instead of the LS problem which \eqref{eq:computeYPI} solves, one can adjust the equation by solving a scaled total least squares (STLS) problem, see \cite{paige2002scaled}.\\
Another potential noisy scenario is when disturbances are injected into the interconnection~\eqref{eq:twosidedInterconnection} between the system and signal generators. Consider for instance the system
\begin{equation} \label{eq:twosidedInterconnectionNoise}
\begin{split} 
     \dot{\omega} = S \omega, \qquad \qquad \,
     \dot{x} = Ax + B(L\omega + v), \,\, \\
     \dot \varpi = Q \varpi + R (C x + z),  \qquad
     d = \varpi + \Upsilon x,
\end{split}
\end{equation}
with $v(t)  \in \RR$ and $z(t) \in \RR$ independent, zero-mean, stationary, white Gaussian noises. The analysis of the estimation of $\Upsilon \Pi$ under these two disturbances is in general complicated, although in Section \ref{sec:NumericalExampleTwo-sided} we present numerical results that indicate a certain degree of robustness in this general case. Here we present an analysis of how the estimation error for $\Upsilon B$ (based on the approach in Theorem~\ref{theorem:computeYB}) is affected by such disturbances, motivated by the fact that the error $\epsilon_{\Upsilon B}^i$ has a direct impact on the estimation error for $\Upsilon \Pi$ (and $(\Upsilon \Pi)^{-1}$ eventually). To this end, consider the swapped interconnection~\eqref{eq:swappedInterconnection} affected by noise, described by
\begin{equation} \label{eq:swappedInterconnectionNoise}
    \dot{x} = Ax + Bu, \quad \dot \varpi = Q \varpi + R (C x + z), \quad d = \varpi + \Upsilon x, 
\end{equation}
with $z(t) \in \RR$ a zero-mean, stationary, white Gaussian noise. Then the time evolution of the signal $d(t)$ is expressed as 
$
    \dot{d} = Qd + \Upsilon Bu + Rz 
$ with $d(0) = 0$. For $u(t)=\delta_0(t)$, the impulse response $d(t) = e^{Qt}\Upsilon B +   \int_0^t e^{Q(t- \tau)} R z(\tau) d\tau$. Then, it can be easily verified that, in contrast to \eqref{eq:errorYB}, the estimation error of estimate~\eqref{eq:computeMomentSingalPoint} under such noise now becomes 
$\hat{\epsilon}_{\Upsilon B}^i = \epsilon_{trans}(t_i) + \epsilon_{noise}(t_i)$ with $\epsilon_{trans}(t_i) =  e^{-Qt_i}\Upsilon e^{At_i} B$ a transient term as in \eqref{eq:errorYB} that exponentially decays to $0$, and $\epsilon_{noise}(t_i) = -\int_0^{t_i}e^{Q \tau} R z(\tau) d\tau$ an additional Brownian-like motion term induced by the noise, with initial condition $\epsilon_{noise}(0) = 0$. As a consequence of this additional term, the stopping condition of an algorithm based on Theorem~\ref{theorem:computeYB} for the estimation of $\Upsilon B$ may not be triggered. Note, however, that
for arbitrary probability $\bar{p} \in (0, 1)$, we can find a sufficiently large $\eta_1$ such that for all $i \geq 0$, $P\left(||\epsilon_{noise}(t_i) - \epsilon_{noise}(t_{i-1})|| \leq \eta_1 (t_i - t_{i-1})\right) \geq \bar{p}$. Then, given a transient tolerance $\eta_2$ such that $||\epsilon_{trans}(t_i) - \epsilon_{trans}(t_{i-1})|| \leq \eta_2 (t_i - t_{i-1})$ for all $t_i \geq \bar{t}$, it follows that, by selecting $\eta_{\Upsilon B} = \eta_1 + \eta_2$, the stopping condition of Algorithm 1 in \cite{mao2022model} can be fulfilled at all $t_i \geq \bar{t}$, with probability greater or equal to $\bar{p}$. The role of $\eta_1$ here is to guarantee the triggering of the stop condition with a certain probability, once the transient objective (\ie \, $||\epsilon_{trans}||$ decreases to a specific threshold) is achieved. 

\section{Numerical Example}
\label{sec:NumericalExampleTwo-sided}
\begin{figure}
    \centering
    \includegraphics[width=0.98\columnwidth]{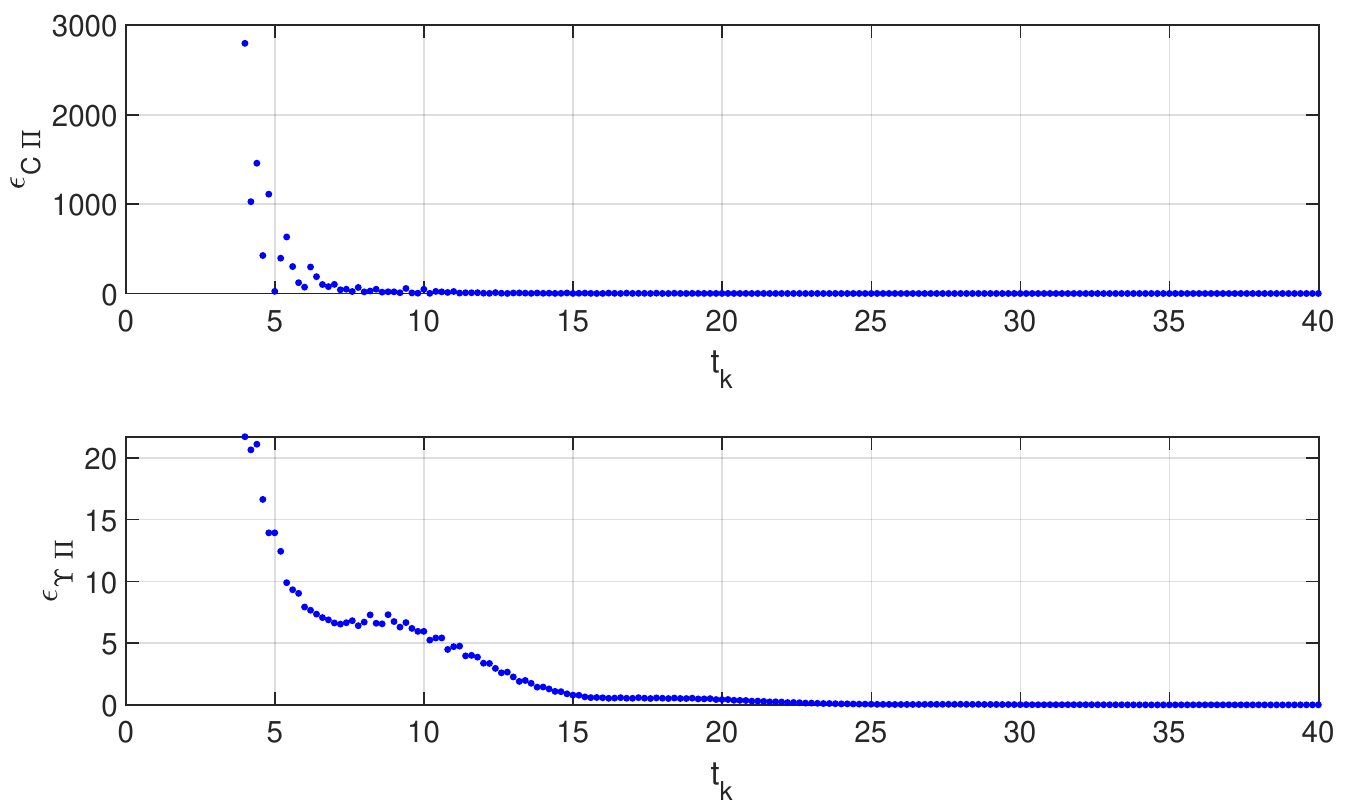}
    \caption{\textbf{Top graph}: normalized error between the matrix $C \Pi$ and the estimation $\widetilde{C \Pi}_k$ over time. \textbf{Bottom graph}: normalized error between the matrix $\Upsilon \Pi$ and the estimation $\widehat{\Upsilon \Pi}_k$ over time.} 
    \label{fig:errorPlotTwoSidedDataDriven}
\end{figure}

\begin{figure}
    \centering
    \includegraphics[trim={2.5cm 9.5cm 2cm 9.5cm},width=0.98\columnwidth]{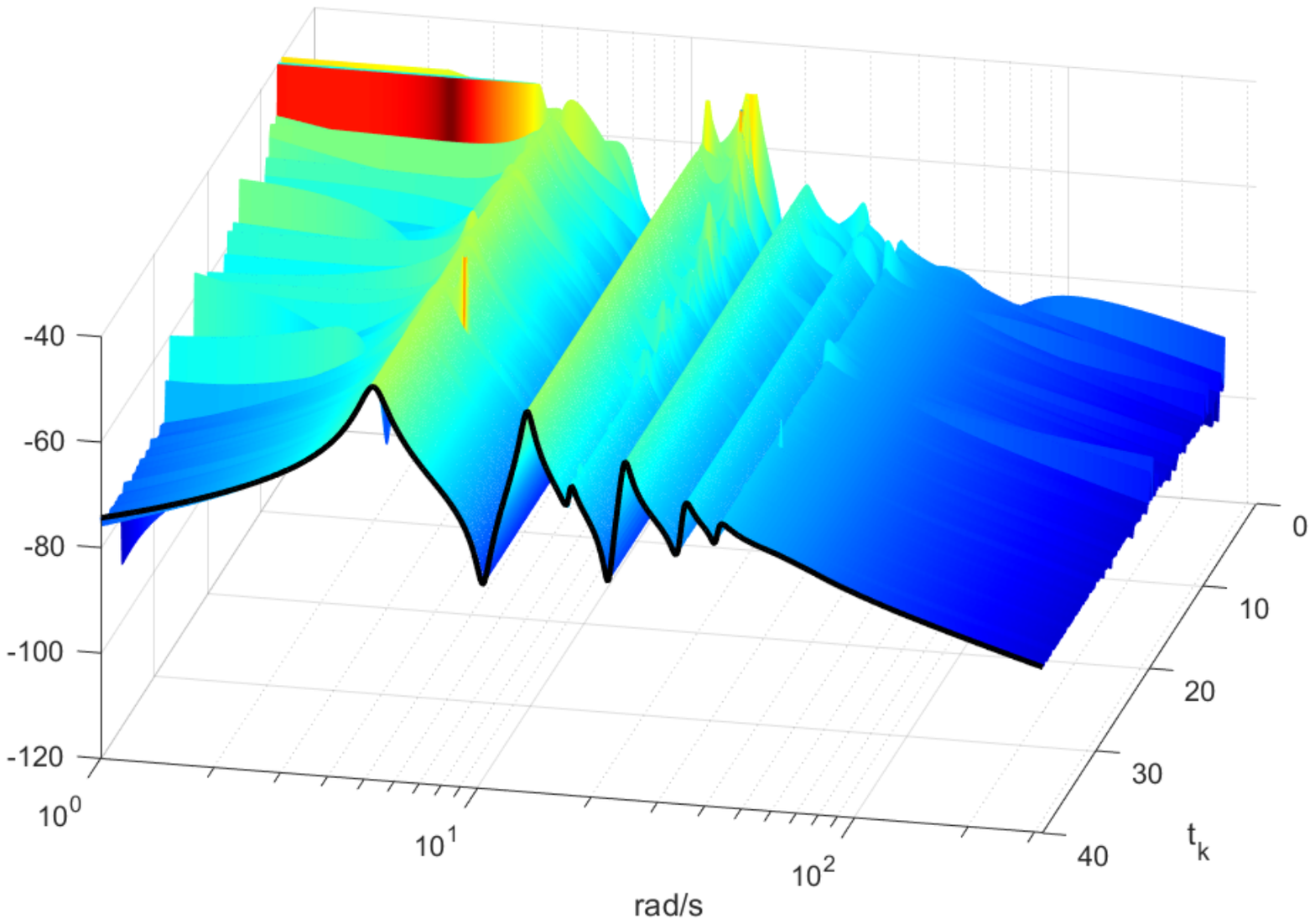}
    \caption{The mesh represents the time-history of the magnitude of the transfer function of the reduced
order model as a function of $t_k$, with $5.99 \le t_k \le 40$ s. The solid/black line indicates
the magnitude of the transfer function of the model obtained using \cite{ionescu2015two}.} 
    \label{fig:evolutionMag}
\end{figure}

\begin{figure}
    \centering
    \includegraphics[width=0.98\columnwidth,height=6.5cm]{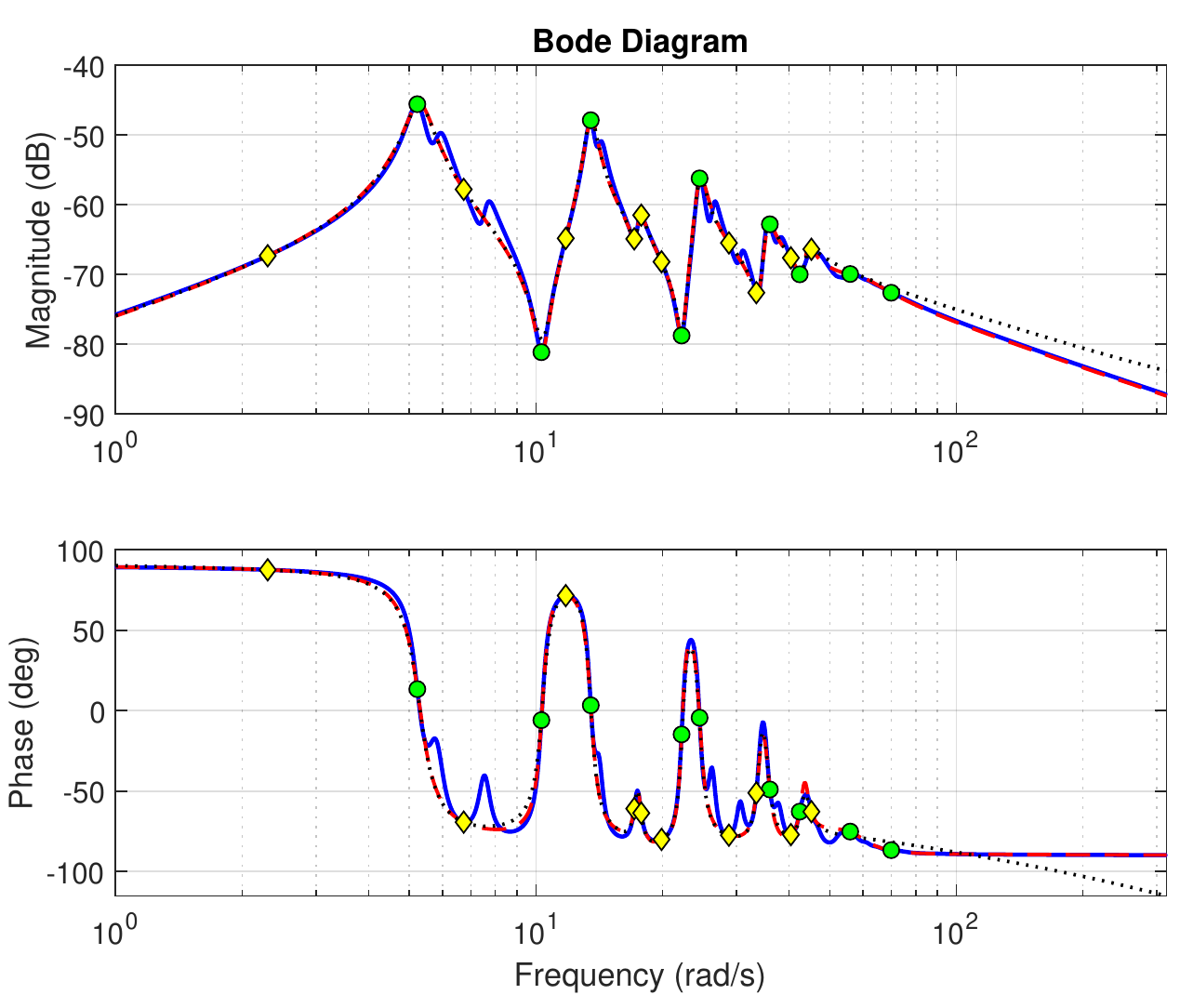}
    \caption{Bode plots of the building model (solid/blue line) and of the reduced-order models in the noise-free case (dashed/red line) and the noisy case (dotted/black line) constructed with the approximations $\widetilde{\Upsilon B}_i$($t_i = 25$ s), $\widetilde{C \Pi}_k$($t_k = 40$ s) and $\widehat{\Upsilon \Pi}_k$($t_k = 40$ s). The coloured markers represent the sets of interpolation points characterized by the matrices $S$ (yellow diamonds) and $Q$ (green circles).} 
    \label{fig:bodePlotTwoSidedDataDriven}
\end{figure}

We consider the \textit{Los Angeles University Hospital}\footnote{The data can be downloaded at \cite{SLICOT}, see \cite{antoulas2000survey}.} building model, a widely-used benchmark model, 
to illustrate the use of Algorithm \ref{alg:On-LineMomentEstimationTwo-sided}. The specifications of this model are summarized as follows: the building has $8$ floors, each with $3$ degrees of freedom, \textit{i.e.}, displacements in $x$- and $y$-directions and rotations; the dynamics of the model can be described by a state-space model of the form (\ref{eq:fullOrderSystem}) with a state of dimension $n = 48$; the model output is the state variable $x_{25}$ which corresponds to motion in the $x$-direction of the first coordinate. In this work, we consider the set of interpolation points characterized by $Q$ as $\pm 0.01 \iota, \pm 5.22 \iota, \pm 10.3 \iota, \pm 13.5 \iota, \pm 22.2 \iota, \pm 24.5 \iota, \pm 36 \iota,$ $\pm 42.4 \iota, \pm 55.9 \iota, \pm 70 \iota$ (which are mostly the major peaks of the \textit{magnitude} plot), whilst the other set characterized by $S$ is
$\pm 2.3 \iota,\pm 19.89 \iota, \pm 11.77 \iota, \pm 6.73 \iota, \pm 17.13 \iota, \pm 17.8 \iota,$ $\pm 28.77 \iota, \pm 40.4 \iota, \pm 33.43 \iota, \pm 45.2 \iota$ (which are mostly the major peaks of the \textit{phase} plot). This selection results in a reduced order of $20$ (for comparison, this benchmark is normally reduced to order 31, see \cite{antoulas2005approximation}).\\
For the moments $\Upsilon B$, using Theorem~\ref{theorem:computeYB}, we obtain a sufficiently accurate estimation $\widetilde{\Upsilon B}_i$ with $||\epsilon_{\Upsilon B}^i|| < 10^{-6}$ at the end of an experiment of $25$ seconds. Based on $\widetilde{\Upsilon B}_i$, we estimate $C \Pi$ and $\Upsilon \Pi$ using Algorithm \ref{alg:On-LineMomentEstimationTwo-sided} with an online experiment of $40$ seconds. The time sequence $\{t_k\}$ is selected equidistantly with sampling time $\Delta t = 0.1$ s. The associated normalized estimation errors over time, namely $\bar{\epsilon}_{C \Pi}(t) = \frac{||\widetilde{C \Pi}_k - C \Pi||}{|| C \Pi||}$ and $\bar{\epsilon}_{\Upsilon \Pi}(t) = \frac{||\widehat{\Upsilon \Pi}_k - \Upsilon \Pi||}{|| \Upsilon \Pi||}$, are shown in Fig. \ref{fig:errorPlotTwoSidedDataDriven}. Observe that these two error signals rapidly decay to near-zero when $t_k > 25$ s.\\
Fig.~\ref{fig:evolutionMag} shows the evolution of the Bode (magnitude) plot of the (approximate) reduced-order model obtained using the estimates $\widetilde{\Upsilon B}_i$, $\widetilde{C \Pi}_k$ and $\widetilde{\Upsilon \Pi}_k$, as a function of $t_k \in [5.99, 40]$ seconds. The phase plot, which is analogous, is omitted for the sake of space. The solid/black line represents the magnitude plot of the reduced-order model constructed with exact $\Upsilon B$, $C \Pi$ and $\Upsilon \Pi$, \textit{i.e.} by applying \cite{ionescu2015two}. The figure shows that the (approximate) reduced-order model gradually approaches the unique two-side moment-matching model as time increases. A comparison with the original model is shown in the next figure: the Bode plots of the building model (solid/blue line) and the resulting reduced-order model (dashed/red line) at $t_k = 40$ s is shown in Fig. \ref{fig:bodePlotTwoSidedDataDriven}. It can be observed that this reduced-order model achieves the two-sided moment matching by simultaneously matching the two sets of interpolation points characterized by the matrices $S$ (yellow diamonds) and $Q$ (green circles) respectively. To further demonstrate the performance of our algorithm under noises, we consider system~\eqref{eq:twosidedInterconnectionNoise} with signal-to-noise ratios for $v(t)$ and $z(t)$ of $\text{SNR}_v = \text{SNR}_z = 20$ dB. The resulting reduced model is also shown in Fig. \ref{fig:bodePlotTwoSidedDataDriven} as the dotted/black line. Observe that this model still provides a sufficiently good matching of moments at the prescribed interpolation points, showing that our algorithm presents a certain level of robustness against the white noise injections considered in Section~\ref{sec:Estimation under Noise Corruption}. 

\section*{Conclusions}
We have presented a time-domain data-driven approach to address the problem of model reduction by two-sided moment matching for linear systems. This approach does not require any knowledge of a state-space representation. Firstly, based on a ``two-sided'' interconnection, an estimation algorithm has been proposed to asymptotically approximate the ``moment-related'' matrices $C\Pi$ and $\Upsilon \Pi$. Then, through leveraging those estimations, the unique $v$-order reduced-order model that matches the moments at $2\nu$ interpolation points has been obtained. An approach that directly estimates the matrix $(\Upsilon \Pi)^{-1}$ has also been proposed. We have also discussed how the algorithm performs against some measurement noises and disturbances that may appear in the interconnection. Finally, we have demonstrated the use of the proposed algorithm by means of a widely-used benchmark model.

\bibliographystyle{ieeetr}
\bibliography{ref}           

\end{document}